\documentclass[onecolumn,12pt]{IEEEtran}

\usepackage{amsmath}
\usepackage{amsfonts}
\usepackage{pifont}
\usepackage{amssymb}
\usepackage{amsthm}
\usepackage{tikz}
\usepackage{graphicx}
\usepackage{array}
\usepackage[small,justification=centering]{caption}

\theoremstyle{plain}

\newtheorem{lemma}{Lemma}
\newtheorem{proposition}{Proposition}
\newtheorem{definition}{Definition}
\newtheorem{theorem}{Theorem}
\newtheorem{corollary}{Corollary}
\theoremstyle{definition}

\renewcommand{\r}[1]{\text{rank$(#1)$}}

\newcommand{\rs}[3]{\text{${#1}^{(#2)}_{#3}$}}

\begin{document}
\title{Repair Locality From a Combinatorial Perspective}

\author{\IEEEauthorblockN{Anyu~Wang and
        Zhifang~Zhang}

\IEEEauthorblockA{Key
Laboratory of Mathematics Mechanization, NCMIS\\
Academy of Mathematics and Systems Science, Chinese Academy of Sciences\\
Beijing, 100190\\
Email: wanganyu@amss.ac.cn,~ zfz@amss.ac.cn}
}
\maketitle
\thispagestyle{empty}

\begin{abstract}

Repair locality is a desirable property for erasure codes in distributed storage systems.
Recently, different structures of local repair groups have been proposed in the definitions of repair locality. In
this paper, the concept of \emph{regenerating set} is introduced to characterize the local repair groups.
A definition of locality $r^{(\delta -1)}$ (i.e., locality $r$ with repair tolerance $\delta -1$) under the most general structure of regenerating sets is given.
All previously studied locality turns out to be special cases of this definition.
Furthermore, three representative concepts of locality proposed before are reinvestigated under the framework of regenerating sets, and their respective upper bounds on the minimum distance are reproved in a uniform and brief form.
Additionally, a more precise distance bound is derived for the \emph{square code} which is a class of linear codes with locality $r^{(2)}$ and high information rate, and  an explicit code construction attaining the optimal distance bound is obtained.
\end{abstract}

\section{Introduction}
In modern large-scale storage systems, erasure codes can afford higher data reliability with considerably smaller storage overhead \cite{erasure_vs_replication}.
An important issue in utilizing erasure codes is data repair in case of node failures so that the whole storage system keeps the same level of redundancy.
Nevertheless how to reduce the repair cost becomes a key problem that affects the practical applications of erasure codes.
There are several cost metrics that can be optimized during the repair process: the repair bandwidth \cite{Network_coding_for_distributed_storage_systems_Dimakis}, i.e., the total number of bits communicated in the network, the number of bits read from existing disks \cite{Optimal_Access_Tamo}, and the repair locality \cite{On_the_locality_of codeword_symbols_Huang, IO_optimal_Huang}, i.e., the number of nodes that participate in the repair process.
Each of these metrics is relevant to different application environments.
For cloud storage applications, the main performance bottleneck is the disk I/O \cite{IO_optimal_Huang}, which is proportional to the number of nodes connected during the repair process.

A related performance metric is {\it repair locality} which was first introduced for linear scalar codes \cite{On_the_locality_of codeword_symbols_Huang,Oggier_Self}. Specifically,  a coordinate of a linear scalar code has locality $r$ if the value at this coordinate can be recovered by a linear combination of the values at $r$ other coordinates. We say these $r$ coordinates form a {\it local repair group} of the former coordinate.
Then in \cite{Loc_repair_codes} the locality $r$ was generalized to vector and nonlinear codes while the structure of local repair groups remained unchanged.
Later, the structure of error-correcting codes was adopted in local repair groups,  which gave the definition of locality $(r,\delta)$ in \cite{r_delta_Prakash2012}. Since the local repair group provides a subcode with minimum distance $\delta$, the locality $(r,\delta)$ can tolerate up to $\delta-1$ erasures, which means even when $\delta-1$ nodes fail in the system, each failed node can still be repaired by accessing $r$ existing nodes.
This locality was also generalized to vector and nonlinear codes in \cite{RankMetric}. Recently, another kind of locality with $(\delta-1)$-erasure talerance has been proposed as the $(r,\delta)_c$-locality \cite{r_delta_c_2013}, where the local repair group consists of $\delta-1$ disjoint subsets. This new structure of local repair groups leads to an improvement in the minimum distance.
With all these definitions of locality, the upper bounds on the minimum distance were derived respectively, and codes attaining the upper bounds were constructed.

In this paper, we introduce the concept of {\it regenerating set} to characterize the local repair groups.
Under the framework of regenerating sets, we develop a uniform approach to analyze the minimum code distance for different kinds of locality.
Specifically, a connection between the minimum distance and the regenerating set structure is established for any code, so the problem of estimating the code distance is transformed into calculating the size of  unions of regenerating sets, and the latter is a simple combinatorial problem.
In detail, this paper includes three contributions that benefit from the framework of regenerating sets.
\begin{itemize}
\item[(1)] {\it The most general definition.}
We define the locality $r^{(\delta-1)}$ to describe the locality $r$ along with repair tolerance $\delta-1$ under a general structure of local repair groups.
The definition applies to both linear and nonlinear codes. All previously studied locality are actually  special cases of this definition.
\item[(2)] {\it Uniform and brief proofs.}
We reinvestigate three representative families of codes  with different locality proposed before, and reprove the upper bounds of the minimum distance in a combinatorial way.
The proofs present an uniform and brief form.
\item[(3)] {\it Precise bound.}
We derive an upper bound on the minimum distance for a class of specific codes. This bound turns out to be more precise than that given before \cite{r_delta_c_2013}.
Moreover, we present an explicit code construction that attains this upper bound.
\end{itemize}

\subsection{Related Work}

\begin{table*}[t]
\centering
\caption{}\label{tabComp}
\renewcommand{\arraystretch}{1.3}
\begin{tabular}{|c|c|c|c|c|}
\hline
locality & $\alpha$ & linear / nonlinear & $\delta$  & local repair group \\ \hline
\cite{On_the_locality_of codeword_symbols_Huang}: locality $r$ & $\alpha =1$ & linear & $\delta =2$ & single subset  \\ \hline
\cite{Loc_repair_codes}: $C(n,r,d,\alpha)$ codes & $\alpha \ge 1$ & both & $\delta = 2$ & single subset \\ \hline
\cite{r_delta_Prakash2012}: locality $(r,\delta)$  & $\alpha = 1$  & linear & $\delta \ge 2$ & error correcting codes \\ \hline
\cite{RankMetric}: $(r,\delta,\alpha)$ codes & $\alpha \ge 1$  & both & $\delta \ge 2$ & error correcting codes  \\ \hline
\cite{r_delta_c_2013}: $(r,\delta)_c$-locality & $ \alpha = 1 $ & linear & $\delta \ge 2$ & disjoint repair sets \\ \hline
\cite{LRC_Alternatives_Oggier2013}: repair tolerance $\delta(i)$ & $\alpha = 1$ & linear & $\delta \ge 2$ & general \\ \hline
this paper: locality $r^{\delta-1}$ & $\alpha \ge 1$ & both & $\delta \ge 2$ & general \\ \hline
\end{tabular}
\end{table*}

As we have stated, previously proposed locality all fall into the scope of our locality $r^{(\delta-1)}$.
Table \ref{tabComp} gives a comparison of different definitions of locality, where $\alpha$ stands for the size of each coded fragment, namely, $\alpha=1$ means the locality only applies to scalar codes while $\alpha \ge 1$ means it also applies to vector codes, and $\delta$ denotes the repair tolerance.

The framework of regenerating sets proposed in this paper extends the matroid approach used in \cite{LRC_Matroid} to the vector case and nonlinear case.
Particularly, it is sufficient for paper \cite{LRC_Matroid} to study circuits in linear matroids because only linear scalar codes were concerned there.
However, because of generalization of the locality $r^{(\delta-1)}$ in this paper we alternatively define the regenerating set to characterize local repair groups, and develop effective approaches accordingly to prove the code distance bound.

\subsection{Organization}
Section II introduces the concept of regenerating set and shows its connection with the minimum distance.
Section III gives the definition of locality $r^{(\delta-1)}$ and reproves the upper bounds of code distance for three kinds of locality proposed before.
Section IV derives an upper code distance bound for the square codes and an explicit  construction attaining this bound. Section V concludes the paper.

\section{Regenerating Sets and The Minimum Distance}\label{secRegSet}

Let $G$ be an encoding function that takes input a file of size $M$ over an alphabet $\Sigma$ and outputs $n$ coded fragments of size $\alpha$ over $\Sigma$, that is,
\begin{equation*}
G(X) = (Y_1,\cdots,Y_n),
\end{equation*}
where $X \in \Sigma^M$ and $Y_i \in \Sigma^\alpha$ for $i = 1,\cdots,n$.
Note that $X$ can be viewed as a random variable which is uniformly drawn from $\Sigma^M$ and $Y_1,\cdots,Y_n$ are random variables over $\Sigma^\alpha$.
Namely, $H(X) = M$ and $H(Y_i) \le \alpha$, where $H(\cdot)$ is the $|\Sigma|$-ary entropy function.
For convenience, we denote the code determined by the encoding function $G$ as an $(n,(M,\alpha),d)$ code $\mathcal{C}$, where $d$ is the minimum distance defined as follows:

\begin{definition}\label{defMinDist}
The minimum distance of $\mathcal{C}$ is defined as
\begin{equation*}
d = n - \max \{ | E | : E \subseteq [n] \text{ and } H(Y_E) < M\},
\end{equation*}
where $[n]$ denotes the set of integers $\{1,2,\cdots,n\}$ and $Y_E$ is the set of random variables $\{Y_i\}_{i \in E}$.
\end{definition}

It follows from the definition that any $n-d+1$ of the variables $Y_1,Y_2,\cdots,Y_n$ have joint entropy $M$, and therefore the $(n,(M,\alpha),d)$ code $\mathcal{C}$ can tolerate up to $d -1$ erasures.
To ensure the repair of all coordinates, we assume $d \ge 2$ throughout the paper.
A trivial result is that $H(Y_1,\cdots,Y_n) = M$.

\subsection{Regenerating Sets}
Now we define the regenerating set with respect to an $(n,(M,\alpha),d)$ code $\mathcal{C}$.
\begin{definition}\label{defRegSet}
  For any $i\in [n]$, a regenerating set of the $i$-th coordinate is a subset $R\subseteq [n]$ satisfying $i\in R$ and $H(Y_i\mid Y_{R\setminus \{i\}})=0$.
\end{definition}
It can be seen that any coordinate of $\mathcal{C}$ has at least one regenerating set when the minimum distance $d\geq2$.
Moreover, if $R$ is a regenerating set of the $i$-th coordinate, then any set $R'$ satisfying $R\subseteq R'\subseteq [n]$ is also a regenerating set of the $i$-th coordinate.
We denote the collection of all regenerating sets of the $i$-th coordinate as $\mathcal{R}_i$.

\begin{definition}
A sequence of regenerating sets $R_1, R_2,...,R_m$, where $R_i\in\mathcal{R}_{l_i}$ for $1\leq i\leq m$ and $l_i\in[n]$, is said to have a nontrivial union if $l_j\not\in\cup_{i=1}^{j-1}R_i$ for $1\leq j\leq m$.
\end{definition}

The structure of nontrivial union plays an important role in estimating the minimum distance of a code.
The following proposition gives an upper bound on the entropy of a nontrivial union of regenerating sets in terms of its set size.
\begin{proposition}\label{propEntSiz}
  Suppose a sequence of regenerating sets $R_1, R_2,...,R_m$ has a nontrivial union, where $R_i\in\mathcal{R}_{l_i}$ and $l_i\in[n]$ for $1\leq i\leq m$. Then $H(Y_{\cup_{i=1}^mR_i})\leq \alpha(|\cup_{i=1}^mR_i|-m)$.
\end{proposition}
\begin{proof}
We prove this by induction on $m$.
First for $m=1$,
\begin{eqnarray*}
H(Y_{R_1}) & =& H(Y_{R_1 \backslash \{l_1\}}) + H(Y_{l_1} |  Y_{R_1 \backslash \{l_1\}}) \\
& = & H(Y_{R_1 \backslash \{l_1\}}) \\
& \le & \alpha(| R_1 | -1).
\end{eqnarray*}

Then suppose the argument holds for $m -1$, where $m >1$.
Let $R_m = \{l_m\} \cup A \cup B$  be a partition of $R_m$ such that $A \subseteq \cup_{i=1}^{m-1} R_i$ and $B \cap (\cup_{i=1}^{m-1} R_i) = \emptyset$.
Because $H(Y_{l_m} | Y_{A \cup B}) = 0$ and $A \subseteq \cup_{i=1}^{m-1} R_i$, it has
\begin{eqnarray*}
H(Y_{\cup_{i=1}^{m} R_i}) &=& H(Y_{\cup_{i=1}^{m-1} R_i}, Y_B) \\
& \le & H(Y_{\cup_{i=1}^{m-1} R_i}) + H(Y_B)\\
& \le & \alpha (| \cup_{i=1}^{m-1} R_i | - (m-1)) + \alpha | B | \\
& = &\alpha (| \cup_{i=1}^{m} R_i | -m),
\end{eqnarray*}
where the last equality comes from the definition of nontrivial union and the partition of $R_m$.
\end{proof}

\subsection{Upper Bound of the Minimum Distance}

We continue to define some notations with respect to an $(n,(M,\alpha),d)$ code and derive an upper bound of $d$.
First, define a function $\Phi(x)$ to be the minimum size of a nontrivial  union of $x$ regenerating sets,
$$\mbox{i.e.,~~}\Phi(x)=\min\{|\cup_{i=1}^xR_i|: R_i\in\mathcal{R}_{l_i} \mbox{~and~ $R_1,...,R_x$ have a nontrivial union}\}\;.$$
In particular, we assume $\Phi(0)=0$. It is easy to see $\Phi(x+1)\geq \Phi(x)+1$, thus $\Phi(x)-x$ is an increasing function with respect to $x$.

Define $$\rho=\max\{x\mid \Phi(x)-x<\frac{M}{\alpha}\}\;.$$
Obviously, $\rho\geq 0$.
The next is a corollary of Proposition \ref{propEntSiz}.
\begin{corollary}\label{corSizOfRegSet}
For $ 0 \le x \le \rho $, let $R^{(x)}_1,R^{(x)}_2,\cdots,R^{(x)}_x$ be a sequence of regenerating sets that has a nontrivial union and $\Phi (x) = | \cup_{i=1}^{x} R^{(x)}_i |$.
Then $[n] - \cup_{i=1}^{x} R^{(x)}_i  \ne \emptyset$.
\end{corollary}
\begin{proof}
By Proposition \ref{propEntSiz},
\begin{eqnarray*}
H(Y_{\cup_{i=1}^{x} R^{(x)}_i }) & \le & \alpha( | \cup_{i=1}^{x} R^{(x)}_i | - x) \\
& = & \alpha(\Phi(x)-x) \\
& \le & \alpha(\Phi(\rho)-\rho) \\
& < & M.
\end{eqnarray*}
Since $H(Y_1,\cdots,Y_n) = M$, then $[n] - \cup_{i=1}^{x} R^{(x)}_i  \ne \emptyset$.
\end{proof}

\begin{theorem}\label{thmMinDis}
  Let $\mathcal{C}$ be an $(n,(M,\alpha),d)$ code, then $$d\leq n-\lceil\frac{M}{\alpha}\rceil+1-\rho\;.$$
\end{theorem}
\begin{proof}
Without loss of generality, suppose $\Phi(\rho) = \left| \cup_{i=1}^{\rho} R_i \right|$, where $R_i \in \mathcal{R}_{l_i}$ and $R_1,\cdots,R_\rho$ have a nontrivial union.
Then
$$\left| \cup_{i=1}^{\rho} R_i \right| = \Phi(\rho) \le  \rho + \lceil \frac{M}{\alpha} \rceil -1$$
by the definition of $\rho$.
From Corollary \ref{corSizOfRegSet}, $[n]- \cup_{i=1}^{\rho} R_i \neq \emptyset$.
Furthermore, for any set $T \subseteq[n] - \cup_{i=1}^{\rho} R_i $ with $ \left| T \cup (\cup_{i=1}^{\rho} R_i) \right| \le \rho +\lceil \frac{M}{\alpha} \rceil -1$, we have
\begin{eqnarray*}
H(Y_{(\cup_{i=1}^{\rho}R_i) \cup T}) & \le & H(Y_{\cup_{i=1}^{\rho} R_i}) + H(Y_T)\\
& \le & \alpha ( \left| \cup_{i=1}^{\rho}R_i \right| - \rho) + \alpha \left| T \right| \\
& = & \alpha (\left|( \cup_{i=1}^{\rho}R_i ) \cup T \right| - \rho) \\
& \le & \alpha (\lceil \frac{M}{\alpha} \rceil -1) \\
& < & M.
\end{eqnarray*}
Particularly, choose a set $T^\prime \subseteq[n] -\cup_{i=1}^{\rho} R_i$ with $ \left| T^\prime \cup (\cup_{i=1}^{\rho} R_i) \right| = \rho +\lceil \frac{M}{\alpha} \rceil -1$, then
$H(Y_{(\cup_{i=1}^{\rho} R_i )\cup T^\prime}) < M.$
Thus by Definition \ref{defMinDist}, $d \le n- \left| (\cup_{i=1}^{\rho}R_i) \cup T^\prime \right| = n- \lceil \frac{M}{\alpha} \rceil+1 - \rho$.
\end{proof}

From the theorem, upper-bounding the minimum distance mainly depends on computing the value of $\rho$ which in turn relies on computation of the function $\Phi(x)$.

\section{Codes with locality}\label{secRepTol}
Next we give the general definition of locality.
It can be regarded as an extension of the repair tolerance defined in \cite{LRC_Alternatives_Oggier2013} to include the vector case and the nonlinear case.
\begin{definition}\label{defLocRepTol}
Let $\mathcal{C}$ be an $(n,(M,\alpha),d)$ code.
For $i \in [n]$, we say the $i$-th coordinate of $\mathcal{C}$ has locality $r$ with repair tolerance $\delta -1$, denoted as locality $r^{(\delta-1)}$, if for all subset $E \subseteq [n]$ containing $i$ with $\left| E \right| \le \delta -1$, there exists a regenerating set $R \in \mathcal{R}_i$ such that
\begin{itemize}
\item[(1)~]$\left| R \right| \le r+1$, and
\item[(2)~]$R \cap E = \{i\}$.
\end{itemize}
\end{definition}

That is, a coordinate of $\mathcal{C}$ has locality $r^{(\delta-1)}$ if for any codeword of $\mathcal{C}$, the value at this coordinate can be regenerated by accessing at most $r$ other coordinates even in the presence of any other $\delta-2$ erasures.
The generalization of our definition of locality $r^{(\delta-1)}$ is twofold.
When $\delta=2$, it coincides with the repair locality $r$ defined for vector codes in \cite{Loc_repair_codes}, and certainly coincides with the repair locality $r$ in \cite{On_the_locality_of codeword_symbols_Huang} if we further restrict $\mathcal{C}$ to a linear scalar code.
When $\delta > 2$, the definition of locality $r^{(\delta-1)}$ describes the repair tolerance of $\delta-1$ erasures in the most general way, instead of specifying the structure of local repair groups that provides the $(\delta-1)$-erasure tolerance.
Therefore, the locality defined in \cite{r_delta_Prakash2012,RankMetric} by using inner-error-correcting code and that in \cite{r_delta_c_2013} by using disjoint repair sets both fall into the scope of our definition.
We call an $(n,(M,\alpha),d)$ code $\mathcal{C}$ has locality $r^{(\delta-1)}$ if for all $i \in [n]$ the $i$-th coordinate of $\mathcal{C}$ has locality $r^{(\delta-1)}$.

In the following we reinvestigate some previously studied locality from a combinatorial perspective.
Namely, we describe the locality by specifying the structure of their regenerating sets and upper-bound the minimum distance by estimating the size of some set unions.

\subsection{The Code $\mathcal{C}(n,r,d,\alpha)$}
As defined in \cite{Loc_repair_codes} the $i$-th coordinate of a code has repair locality $r$ if the value at this coordinate is a function of values at $r$ other coordinates.
The notation $ C(n,r,d,\alpha)$ is used there to denote  a code with all symbol locality $r$.
By using the concept of regenerating sets, the code $\mathcal{C}(n,r,d,\alpha)$ is an $(n,(M,\alpha),d)$ code satisfying that for all $i\in [n]$ there exists a set $R_i\in\mathcal{R}_i$ with $|R_i|\leq r+1$.

\begin{lemma}\label{lemPhiLocR}
 For a code $\mathcal{C}(n,r,d,\alpha)$, it holds that $\Phi(x)\leq (r+1)x$, where $0\leq x\leq \rho+1$.
\end{lemma}
\begin{proof}
 We prove this lemma by induction on $x$. Because $\Phi(0)=0$, the lemma trivially holds for $x=0$.
 Assume it holds for $x $, where $x \le \rho$.
 Let $T_x$ be the union of a sequence of $x$ regenerating sets that has a nontrivial union and $\Phi(x) = |T_x| \le (r+1)x$.
 From Corollary \ref{corSizOfRegSet}, $[n] - T_x \ne \emptyset$.
 It follows that there exists $h \in [n] - T_x$ and $R \in \mathcal{R}_h$ with $| R | \le r+1$, therefore
  \begin{eqnarray*}
  \Phi(x+1) & \le & | T_x \cup R | \\
  & \le & | T_x | + | R | \\
  & \le & (r+1)(x+1).
  \end{eqnarray*}
\end{proof}

\begin{theorem}\label{thmDisLocR}
  For a code $\mathcal{C}(n,r,d,\alpha)$, it has
  $$d\leq n-\lceil\frac{M}{\alpha}\rceil-\lceil\frac{M}{r\alpha}\rceil+2\;.$$
\end{theorem}
\begin{proof}
By the definition of $\rho$, $\frac{M}{\alpha} \le \Phi(\rho+1) - (\rho+1)$, and $\Phi(\rho+1) \le (r+1)(\rho+1)$ from Lemma \ref{lemPhiLocR}.
It follows that
\begin{eqnarray*}
\frac{M}{\alpha} & \le & \Phi(\rho+1) - (\rho+1) \\
& \le & (r+1)(\rho+1) - (\rho+1) \\
& = & r(\rho+1),
\end{eqnarray*}
and therefore $\rho \ge \lceil \frac{M}{r \alpha} \rceil -1$.
Consequently, $d\leq n-\lceil\frac{M}{\alpha}\rceil-\lceil\frac{M}{r\alpha}\rceil+2$ by Theorem \ref{thmMinDis}.
\end{proof}

\subsection{The $(n,r,\delta,\alpha)$ Locally Repairable Code}
The $(n,r,\delta,\alpha)$ locally repairable code defined in \cite{RankMetric} is a generalization of the $(r,\delta)$ locality which was first proposed in \cite{On_the_locality_of codeword_symbols_Huang}.
This locality is due to a subcode of length no more than $r+\delta-1$ and minimum distance at least $\delta$.
In other words, an $(n,r,\delta,\alpha)$ locally repairable code is an $(n,(M,\alpha),d)$ code such that for $1\leq i\leq n$, there exists a subset $S_i\subseteq [n]$ satisfying
\begin{itemize}
  \item[(1)~]$i\in S_i$, $\delta \le |S_i|\leq r+\delta-1$; and
  \item[(2)~]For any $E\subseteq S_i$ with $|E|=\delta-1$, and for any $j\in E$, it has $(S_i-E)\cup\{j\}\in\mathcal{R}_j$.
\end{itemize}

\begin{lemma}\label{lemPhiRDeltaLoc}
  For an $(n,r,\delta,\alpha)$ locally repairable code, it holds that $\Phi(x)\leq r\lceil\frac{x}{\delta-1}\rceil+x$, where $0\leq x\leq \rho+1$.
\end{lemma}
\begin{proof}
We prove this lemma by induction on $x$.
First, it trivially holds for $x=0$.
Suppose it holds for $x \le x_0$, where $0\le x_0 \le \rho$.
Denote $x_0+1 = a (\delta-1) + b$ where $a \in \mathbb{Z}$ and $b\in [\delta-1]$.
Let $T_{a(\delta-1)} = R_1 \cup \cdots \cup R_{a(\delta-1)}$ be a nontrivial union of $a(\delta-1)$ regenerating sets such that $\Phi(a(\delta-1)) = |T_{a(\delta-1)}|$. There are two cases:

(1) There exists $ h \in [n] - T_{a(\delta-1)}$ such that $|S_h - T_{a(\delta-1)}| \geq \delta-1$, where the notation $S_h$ comes from the description before this lemma.
Choose $E \subseteq S_h - T_{a(\delta-1)}$ with $|E| = \delta-1$.
Suppose $E = \{i_1,\cdots,i_{\delta-1}\}$.
Let $R_{i_j} = (S_h - E) \cup \{i_j\}$ for $j \in[\delta-1]$.
Then $R_{i_j}\in\mathcal{R}_{i_j}$ and $(\cup_{j=1}^{a(\delta-1)}R_j) \cup (\cup_{j=1}^{b} R_{i_j})$ is a nontrivial union. It follows that
\begin{eqnarray*}
\Phi(x_0+1)& \le& |T_{a(\delta-1)} \cup R_{i_1} \cup \cdots \cup R_{i_b}| \\
& \le & \Phi(a(\delta-1)) + |S_h -E|+b \\
& \le & ar+a(\delta-1)+r+b\\
& = & r \lceil \frac{x_0+1}{\delta-1} \rceil +x_0+1.
\end{eqnarray*}

(2) For any $h \in [n] - T_{a(\delta-1)}$, $|S_h - T_{a(\delta-1)}| < \delta-1$. Define $R_h = (S_h \cap T_{a(\delta-1)}) \cup \{h\}$, then
$R_h \in \mathcal{R}_h$.
If $n - |T_{a(\delta-1)}| \ge b$, then choose $ h_1,\cdots,h_b \in [n] - T_{a(\delta-1)}$.
So
\begin{eqnarray*}
\Phi(x_0+1) & \le & |T_{a(\delta-1)}\cup R_{h_1} \cup \cdots \cup R_{h_b}| \\
& = & |T_{a(\delta-1)}|+b \\
& = & \Phi(a(\delta-1)) +b \\
& \le & r \lceil \frac{x_0+1}{\delta-1} \rceil +x_0+1.
\end{eqnarray*}
If $n - |T_{a(\delta-1)}| < b$, then
$$\Phi(x_0+1)\leq n < |T_{a(\delta-1)}|+b\le  r \lceil \frac{x_0+1}{\delta-1} \rceil +x_0+1.$$

\end{proof}

\begin{theorem}
  For an $(n,r,\delta,\alpha)$ locally repairable code, it has
  $$d\leq n-\lceil\frac{M}{\alpha}\rceil+1-(\lceil\frac{M}{r\alpha}\rceil-1)(\delta-1)\;.$$
\end{theorem}
\begin{proof}
Similar to the proof of Theorem \ref{thmDisLocR}, we have
$$\frac{M}{\alpha} \le \Phi(\rho+1) - (\rho+1)  \le r \lceil \frac{\rho+1}{\delta-1} \rceil.$$
It follows that $\lceil \frac{M}{r\alpha} \rceil \le \lceil \frac{\rho+1}{\delta-1} \rceil$, and therefore $(\lceil\frac{M}{r\alpha}\rceil-1)(\delta-1) \le (\lceil \frac{\rho+1}{\delta-1} \rceil -1)(\delta-1) \le \rho$. Then Theorem \ref{thmMinDis} gives the desired bound.
\end{proof}

\subsection{The $(r,\delta)_c$-Locality}

An $(n,(M,\alpha),d)$ code has $(r,\delta)_c$-locality if for $1\leq i\leq n$, there exist
$R_{i,1},R_{i,2},...,R_{i,\delta-1}\in\mathcal{R}_i$ satisfying
\begin{itemize}
  \item[(1)~]$|R_{i,j}|\leq r+1$ for $1\leq j\leq \delta-1$; and
  \item[(2)~]$R_{i,j}\bigcap R_{i,j'}=\{i\}$ for $1\leq j\neq j'\leq \delta-1$.
\end{itemize}

Paper \cite{r_delta_c_2013} considered the $(r,\delta)_c$-locality only for the linear scalar case, so in the following we set $\alpha=1$ and consider linear codes.
\begin{lemma}\label{lemPhiRDeltaCLoc}
  For a linear $(n,(M,1),d)$ code with $(r,\delta)_c$-locality, it holds $\Phi(x)\leq rx+\lceil\frac{x}{\delta-1}\rceil$ where $0\leq x\leq \rho+1$.
\end{lemma}
\begin{proof}
This lemma is proved by induction on $x$.
First, it trivially holds for $x=0$.
Suppose it holds for $x \le x_0$, where $0\le x_0 \le \rho$.
Denote $x_0+1 = a (\delta-1) + b$ where $a \in \mathbb{Z}$ and $b\in [\delta-1]$.
Let $T_{a(\delta-1)} = R_1 \cup \cdots \cup R_{a(\delta-1)}$ be a nontrivial union of $a(\delta-1)$ regenerating sets such that $\Phi(a(\delta-1)) = |T_{a(\delta-1)}|$. There are two cases:

(1) There exists $ h \in [n] - T_{a(\delta-1)}$ such that $R_{h,j} \cap T_{a(\delta-1)} = \emptyset$ for $j \in [\delta-1]$.
Because of linearity, for $1 \le j \le \delta-1$, there exists $i_j \in R_{h,j} - \{h\}$ such that $R_{h,j} \in \mathcal{R}_{i_j}$.
Then $T_{a(\delta-1)} \cup R_{h,1} \cup \cdots \cup R_{h,b}$ is a nontrivial union.
It follows that
\begin{align*}
\Phi(x_0+1) & \le |T_{a(\delta-1)} \cup R_{h,1} \cup \cdots \cup R_{h,b}| \\
& \le \Phi(a(\delta-1)) + |R_{h,1} \cup \cdots \cup R_{h,b}| \\
& \le ra(\delta-1)+a + rb+1 \\
& = r(x_0+1) + \lceil \frac{x_0+1}{\delta-1} \rceil.
\end{align*}

(2) For any $h \in [n] - T_{a(\delta-1)}$, there exists $j_h \in [\delta-1]$ such that $R_{h,j_h} \cap T_{a(\delta-1)} \ne \emptyset$.
If $n \ge |T_{a(\delta-1)}|+br$, then there exists $h_1,\cdots,h_b$ such that
$$ h_l \in [n] - (T_{a(\delta-1)} \cup R_{h_1,j_{h_1}} \cup \cdots \cup R_{h_1,j_{h_{l-1}}}), \text{ for } 1 \le l \le b $$
because $|T_{a(\delta-1)} \cup R_{h_1,j_{h_1}} \cup \cdots \cup R_{h_1,j_{h_{l-1}}}| \le |T_{a(\delta-1)}| + (l-1)r < n$.
Therefore $R_{h,j_{h}} \in \mathcal{R}_{h}$ for $h \in \{h_1,\cdots,h_b\}$ and $T_{a(\delta-1)} \cup R_{h_1,j_{h_1}} \cup \cdots \cup R_{h_1,j_{h_b}}$ is a nontrivial union.
It follows that
\begin{align*}
\Phi(x_0+1) & \le |T_{a(\delta-1)} \cup R_{h_1,j_{h_1}} \cup \cdots \cup R_{h_1,j_{h_b}}| \\
& \le |T_{a(\delta-1)}| + rb \\
& \le ra(\delta-1)+a + rb \\
& < r(x_0+1) + \lceil \frac{x_0+1}{\delta-1} \rceil.
\end{align*}
If $n < |T_{a(\delta-1)}|+br$, then
\begin{align*}
\Phi(x_0+1) & \le n < |T_{a(\delta-1)}|+rb \\
& < r(x_0+1) + \lceil \frac{x_0+1}{\delta-1} \rceil.
\end{align*}
\end{proof}

\begin{theorem}\label{thmDisRDeltaC}
For a linear $(n,(M,1),d)$ code with $(r,\delta)_c$-locality, it has
\begin{equation*}
d \le n - M +1 - \mu,
\end{equation*}
where $\mu = \lceil \frac{(M-1)(\delta-1)+1}{(r-1)(\delta-1)+1} \rceil -1$.
\end{theorem}
\begin{proof}
By Lemma \ref{lemPhiRDeltaCLoc} and the maximality of $\rho$,
\begin{eqnarray*}
M & \le & \Phi(\rho +1) - (\rho+1)\\
& \le & (r-1) (\rho+1) + \lceil \frac{\rho+1}{\delta-1} \rceil \\
& \le & (r-1)(\rho +1) + \frac{\rho}{\delta-1}+1.
\end{eqnarray*}
It follows that
$\rho \ge \lceil \frac{(M-1)(\delta -1)+1}{(r-1)(\delta-1)+1}\rceil -1.$
Then Theorem \ref{thmMinDis} gives the desired bound.
\end{proof}

\section{The Square Code}
For explicit code constructions, especially when the structure of regenerating sets is given, Theorem \ref{thmMinDis} can be used to give more precise characterization of the minimum distance.
For instance, in this section we utilize Theorem \ref{thmMinDis} to derive a tight bound on the minimum distance of the square code which was introduced in \cite{r_delta_c_2013} as a class of code with $(r,\delta)_c$-locality, a special case of the locality $r^{(\delta-1)}$.
Besides the property of repair tolerance $\delta=3$ for all coordinates, the square code also has the following advantage.
\begin{itemize}
\item[(1)~]High information rate. Some square codes have information rate close to $1$.
\item[(2)~]Desirable code distance. It was shown in \cite{r_delta_c_2013}, under the same level of local repair tolerance and information rate, the square code has the minimum distance beyond the upper bound for the $(r,\delta)$ locality defined in \cite{r_delta_Prakash2012}.
\end{itemize}

We first restate the square code as a linear $(n,(M,1),d)$ code over $\mathbb{F}_q$, where $n  = (r+1)^2$, $ r+1 \le M \le r^2$ and its generator matrix $G =(x_{i,j})_{1\le i,j\le r+1}$ is composed of $n$ column vectors $ x_{i,j}\in \mathbb{F}_q^M $ satisfying
\begin{equation}\label{eq_square_condition}
\begin{cases}
\sum_{i=1}^{r+1} x_{i,j} = 0, \text{ for } 1\le j \le r+1 \\
\sum_{j=1}^{r+1} x_{i,j} = 0, \text{ for } 1 \le i \le r+1.
\end{cases}
\end{equation}

There is a grid corresponding to $\{x_{i,j}\}_{1\le i,j \le r+1}$.
As in Fig. \nolinebreak \ref{fig_square_code}, the vector $x_{i,j}$ stands for the cross point of the $i$-th row and the $j$-th column in the grid.
The sum of all $r+1$ vectors in the same row (or the same column) is zero.
Then for the coordinate $(i,j)$ of $\mathcal{C}$ where $1 \le i,j \le r+1$,
$$ \rs{R}{i,j}{row} = \{(i,1),(i,2),\cdots,(i,r+1)\} \text{ and } \rs{R}{i,j}{col}= \{(1,j),(2,j),\cdots,(r+1,j)\}$$
are its two regenerating sets.
Thus $\mathcal{C}$ has $(r,\delta = 3)_c$-locality, and therefore has locality $r^{(2)}$.

\begin{figure}[htbp]
\centering
\begin{tikzpicture}[scale=0.68]
\filldraw [blue]
(0,0) circle (2pt)
(0,1) circle (2pt)
(1,0) circle (2pt)
(1,1) circle (2pt)
(0,4) circle (2pt)
(0,5) circle (2pt)
(1,4) circle (2pt)
(1,5) circle (2pt)
(4,0) circle (2pt)
(5,0) circle (2pt)
(4,1) circle (2pt)
(5,1) circle (2pt)
(4,4) circle (2pt)
(5,4) circle (2pt)
(4,5) circle (2pt)
(5,5) circle (2pt);

\draw
(0,0) -- (0,1.7)
(1,0) -- (1,1.7)
(4,0) -- (4,1.7)
(5,0) -- (5,1.7)
(0,3.3) -- (0,5)
(1,3.3) -- (1,5)
(4,3.3) -- (4,5)
(5,3.3) -- (5,5)
(0,0) -- (1.7,0)
(0,1) -- (1.7,1)
(0,4) -- (1.7,4)
(0,5) -- (1.7,5)
(3.3,0) -- (5,0)
(3.3,1) -- (5,1)
(3.3,4) -- (5,4)
(3.3,5) -- (5,5);
\draw [dashed]
(0,1.7) -- (0,3.3)
(1,1.7) -- (1,3.3)
(4,1.7) -- (4,3.3)
(5,1.7) -- (5,3.3)
(1.7,0) -- (3.3,0)
(1.7,1) -- (3.3,1)
(1.7,4) -- (3.3,4)
(1.7,5) -- (3.3,5);

\node [below left] at (0,0) {\small $x_{r+1,1}$};
\node [below] at (1,0) {\small $x_{r+1,2}$};
\node [below] at (4,0) {\small $x_{r+1,r}$};
\node [below right] at (5,0) {\small $x_{r+1,r+1}$};

\node [left] at (0,1) {\small $x_{r,1}$};
\node [above right] at (1,1) {\small $x_{r,2}$};
\node [above right] at (4,1) {\small $x_{r,r}$};
\node [right] at (5,1) {\small $x_{r,r+1}$};

\node [left] at (0,4) {\small $x_{2,1}$};
\node [above right] at (1,4) {\small $x_{2,2}$};
\node [above right] at (4,4) {\small $x_{2,r}$};
\node [right] at (5,4) {\small $x_{2,r+1}$};

\node [above left] at (0,5) {\small $x_{1,1}$};
\node [above] at (1,5) {\small $x_{1,2}$};
\node [above] at (4,5) {\small $x_{1,r}$};
\node [above right] at (5,5) {\small $x_{1,r+1}$};

\end{tikzpicture}
\caption{The grid corresponding to vectors $\{x_{i,j}\}_{1 \le i,j \le r+1}$.}
\label{fig_square_code}
\end{figure}
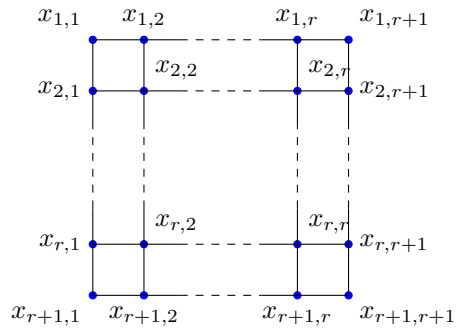

Next, we prove an upper bound on the minimum distance by using Theorem \ref{thmMinDis}.
\begin{theorem}\label{thmSquareUpperBound}
The minimum distance of a square code satisfies
$$ d \le n-M+1 - s,$$
where $s = \max \{ x | g(x)  <M\}$ and
\begin{equation*}
g(x) = \begin{cases}
x r - \frac{x^2}{4}, \text{ if } 2 \mid x \\
x r - \frac{x^2 -1}{4}, \text{ if } 2 \nmid x
\end{cases}
\end{equation*}
is a function defined over all integers $x$ in the range $[0,2r+1]$.
\end{theorem}
\begin{IEEEproof}
First, we prove that $\Phi(x) \le g(x) +x$ for $0\le x \le 2r+1$.
In fact, observe that the $2r+1$ regenerating sets
$$ \rs{R}{1,r+1}{row}, \rs{R}{r+1,1}{col},\rs{R}{2,r+1}{row},\rs{R}{r+1,2}{col},\cdots,\rs{R}{r,r+1}{row},\rs{R}{r+1,r}{col},\rs{R}{r+1,r+1}{row} $$
have a nontrivial union (with respect to the order above).
Thus for $0 \le x \le 2r+1$, $\Phi(x)$ is no more than the size of the first $x$ regenerating sets' union which equals the function value $g(x)+x$.
Consequently, we have $\Phi(x) \le g(x)+x$.

Assume $s \ge \rho+1$, then by the definition of $\rho$ and the increasing property of $\Phi(x)$, it follows that
$$M \le \Phi(\rho+1) - (\rho+1) \le  \Phi(s) - s  \le  g(s),$$
which contradicts the definition of $s$. Therefore $\rho\geq s$ and Theorem \ref{thmMinDis} gives the desired bound.
\end{IEEEproof}

In particular, the square code has $(r,\delta)_c$-locality, so it also satisfies the upper bound in Theorem \ref{thmDisRDeltaC}. But a comparison shows the bound in Theorem \ref{thmSquareUpperBound} is more precise for the square code in general. As an example, Fig. \ref{fig_comparing} displays the two bounds for the square code with $r=5$.

\begin{figure}[htbp]
\centering
\includegraphics[width = 0.618 \textwidth]{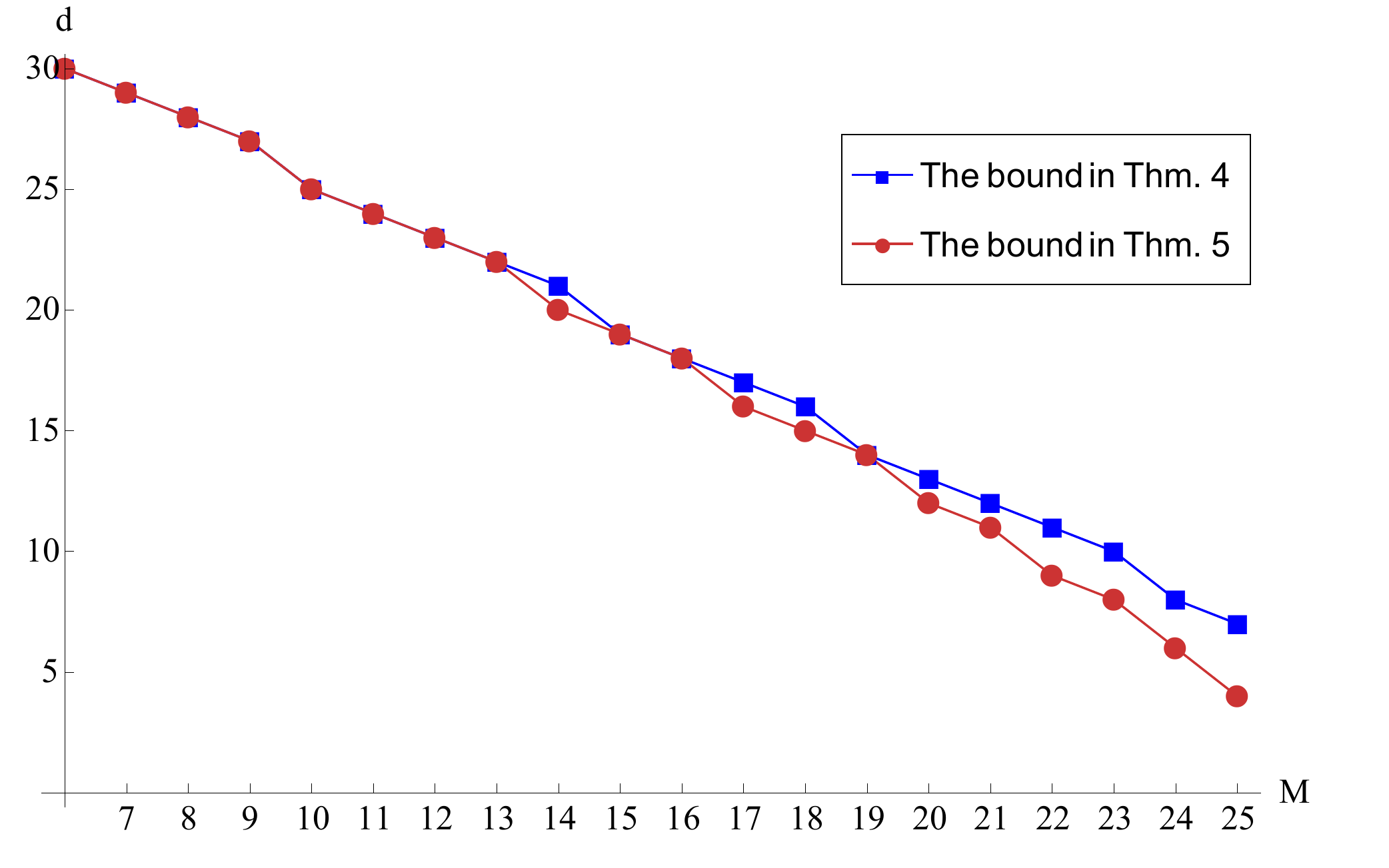}
\caption{The minimum distance upper bound for square codes with $r=5$.}
\label{fig_comparing}
\end{figure}

\subsection{Construction of square codes with optimal distance}
We present an explicit construction of the square code that has the minimum distance $d = n-M+1-s$, showing tightness of the upper bound given in Theorem \ref{thmSquareUpperBound}.

Let $\mathbb{F}_{q^m}$ be an extension field of $\mathbb{F}_q$, where $m \ge r^2$.
Note that $\mathbb{F}_{q^m}$ can be regarded as an $m$-dimensional linear space over $\mathbb{F}_q$.
Then there exist $r^2$ elements $\{\beta_{i,j}\}_{1 \le i,j \le r}$ in $\mathbb{F}_{q^m}$ that are linearly independent over $\mathbb{F}_q$.
Moreover, let
\begin{equation*}
\beta_{r+1,j} = - \sum_{i=1}^{r} \beta_{i,j} \text{ for } 1\le j \le r
\end{equation*}
and
\begin{equation*}
\beta_{i,r+1} = - \sum_{j=1}^{r} \beta_{i,j} \text{ for } 1\le i \le r+1.
\end{equation*}

Let $\mathcal{C}$ be an $(n,(M,1),d)$ linear code over $\mathbb{F}_{q^m}$ with generator matrix $G= (g_{i,j})_{1 \le i,j\le r+1}$, where $n = (r+1)^2,r+1 \le M \le r^2$ and
\begin{equation*}
g_{i,j} = \begin{pmatrix} \beta_{i,j} \\ \beta_{i,j}^q \\ \vdots \\ \beta_{i,j}^{q^{M-1}} \end{pmatrix}.
\end{equation*}
Then $\mathcal{C}$ is a square code of locality $r^{(2)}$ because
\begin{equation*}
\sum_{i=1}^{r+1} g_{i,j} = \begin{pmatrix} \sum_{i=1}^{r+1} \beta_{i,j} \\ \vdots \\ \sum_{i=1}^{r+1} \beta_{i,j}^{q^{M-1}} \end{pmatrix}
= \begin{pmatrix} \sum_{i=1}^{r+1} \beta_{i,j} \\ \vdots \\ (\sum_{i=1}^{r+1} \beta_{i,j})^{q^{M-1}} \end{pmatrix} = 0, \text{ for } 1 \le j \le r+1
\end{equation*}
and similarly, $ \sum_{j=1}^{r+1} g_{i,j} = 0 $ for $1 \le i \le r+1$.

Next we show the minimum distance of $\mathcal{C}$ satisfies $d \ge n -M+1 -s$.
Firstly, we quote a basic result of finite fields.
\begin{lemma}[\upshape \cite{finiteField}]\label{lem_finite_field}
Suppose $x_1,\cdots,x_M \in \mathbb{F}_{q^m}$ are linearly independent over $\mathbb{F}_q$, then
\begin{equation*}
\det
\begin{pmatrix}
x_1 & x_2 & \cdots & x_M \\
x_1^q & x_2^q & \cdots & x_M^q \\
\vdots & \vdots & \ddots & \vdots\\
x_1^{q^{M-1}} & x_2^{q^{M-1}} & \cdots & x_M^{q^{M-1}}
\end{pmatrix}
\neq 0.
\end{equation*}

\end{lemma}
Let $S_1 , \cdots , S_{r+1}$ be a partition of $\{(i,j)\}_{1\le i ,j \le r+1}$, where $S_i = \{ (i,j)\}_{1\le j \le r+1}$ for $1 \le i \le r+1$.
\begin{lemma}\label{lem_suqare_code_generator_matrix}
Suppose $X$ is a subset of $\{(i,j)\}_{1\le i ,j \le r+1}$ such that
\begin{itemize}
\item[(1)] $\left| X \right| \ge M $ and $  \left| X \cap S_i \right| \le r$ for all $1 \le i \le r+1$.
\item[(2)] there exists $1\leq i_0\leq r+1$ such that $X \cap S_{i_0}= \emptyset$.
\end{itemize}
Then $\r{G|_X} = M$.
\end{lemma}
\begin{IEEEproof}
The proof is based on Lemma \ref{lem_finite_field}.
Let $X^\prime$ be a subset of $X$ with size $M$.
It is clear $X^\prime$ also satisfies the condition (1) and (2) in Lemma \ref{lem_suqare_code_generator_matrix}.
Next, we prove that $\r{G|_{X^\prime}} = M$.
By Lemma \ref{lem_finite_field}, it suffices to show the $M$ elements $\{\beta_{i,j}\}_{(i,j) \in X^\prime}$ are linearly independent over $\mathbb{F}_q$.

For $1 \le i \le r+1$, let $V_i$ be the linear space spanned by $\{\beta_{i,j}\}_{(i,j) \in S_i}$ over $\mathbb{F}_q$ and let $V$ be the space spanned by $\{\beta_{i,j}\}_{1 \le i , j \le r+1}$ over $\mathbb{F}_q$.
Because of the choice of $\beta_{i,j}$, the sum of any $r$ out of the $r+1$ spaces $\{V_i\}_{1 \le i \le r+1}$ is equal to $V$.
Note that $\dim(V) = r^2$ and $\dim(V_i) = r$ for $1\le i \le r+1$.
It follows that $V$ is the direct sum of any $r$ subspaces out of $\{V_i\}_{1 \le i \le r+1}$.
Particularly,
\begin{equation}\label{eq_direct_sum}
V = \bigoplus_{\substack{i=1 \\ i \neq i_0}}^{r+1} V_i.
\end{equation}

On the other hand, for $1 \le i \neq i_0 \le r+1$, the condition $ |X\cap S_i|\leq r$ implies that $\{\beta_{i,j}\}_{(i,j)\in X \cap S_i}$ are linearly independent over $\mathbb{F}_q$.
Therefore $\{\beta_{i,j}\}_{(i,j) \in X^\prime}$ are linearly independent over $\mathbb{F}_q$.
\end{IEEEproof}

\begin{theorem}\label{thm_square_code_construction}
Let $\mathcal{C}$ be the square code defined by the generator matrix $G= (g_{i,j})_{1 \le i,j\le r+1}$.
Then
$$d \ge n-M+1-s.$$
\end{theorem}
\begin{IEEEproof}
Assume on the contrary that $d \le n - M -s$.
Then there is a subset $N \subseteq \{(i,j)\}_{1\le i ,j \le r+1}$ such that $| N| = M+s$ and $\r{G|_N} < M$.

Let $N \cap S_i = N_i$, then $N = N_1 \cup \cdots \cup N_{r+1}$ is a partition of $N$.
Suppose that $a = \min \{\left| N_i \right| : 1\le i \le r+1\}$ and $b$ is the number of $N_i$'s such that $\left| N_i \right| = r+1$.
Then clearly $a+b < 2r+1$ and
\begin{eqnarray}\label{eq_abrMs}
\left| N \right| & = & M+s \nonumber \\
& = & \left| N_1 \right| + \cdots + \left| N_{r+1} \right|\nonumber \\
& \ge & (r+1)b + (r+1 -b)a \nonumber \\
& = & (r+1)(a+b) - ab.
\end{eqnarray}

We claim that $s \ge a+b$ because otherwise, it has $s+1 \le a+b < 2r+1$ which leads to
\begin{eqnarray*}
g(s+1) & \le & g(a+b) \\
& \le & r(a+b) - ab \\
& \le & M+s - (a+b) \\
& \le & M-1,
\end{eqnarray*}
where the first inequality is due to increasing property of the function $g(x)$, the second comes from the fact that
\begin{equation*}
ab \le \begin{cases} \frac{(a+b)^2}{4}, \text{ if } a+b \text{ even } \\ \frac{(a+b)^2-1}{4},\text{ if } a+b \text{ odd}, \end{cases}
\end{equation*}
the third is from (\ref{eq_abrMs}) and the last is from the assumption $s+1\leq a+b$.
But $g(s+1)\leq M-1$ contradicts the definition of $s$.
Therefore, $s\ge a+b$.

Suppose the $b$ sets of size $r+1$ are $ N_{i_1} ,\cdots ,  N_{i_b} $ and $N_{i_0} $ is a set of size $a$.
Then delete one element in each of $ N_{i_1}, \cdots, N_{i_b} $ and further delete $N_{i_0}$ from $N$, we get a subset of $N$, denoted as $\tilde{N}$.
It is clear that $| \tilde{N} \cap S_i | \le r$ for $1\le i \le r+1$ and $\tilde{N} \cap S_{i_0}= \emptyset$.
Additionally, $| \tilde{N} | = \left| N \right| -(a+b) \ge M$ because $|N|=M+s$ and $s\geq a+b$.
By Lemma \ref{lem_suqare_code_generator_matrix}, $\r{G|_N} \ge \r{G|_{\tilde{N}}} =M$, which contradicts the choice of $N$.

\end{IEEEproof}

\section{Conclusion}
We introduce the regenerating set which can be used to characterize the local repair groups of any locally repairable codes.
A connection between the regenerating set and the minimum distance is established.
Under the framework of regenerating sets, we derive more general definition, more uniform and brief proofs, and more precise bound. This framework are expected to provide deeper insight into the design of locally repairable codes.

\end{document}